\documentclass[12pt]{iopart}

\usepackage{graphicx, amssymb, amsthm, epsfig, amsbsy, amsfonts, setstack}

\newtheorem{theorem}{Theorem}
\newtheorem{lemma}{Lemma}

\theoremstyle{definition}

\newcommand{\ui}{\textrm{i}}
\newcommand{\ue}{\textrm{e}}
\newcommand{\UI}{\mathbf{I}}

\newcommand{\ud}{\mathrm{d}}

\newcommand{\bbN}{{\mathbb N}}

\newcommand{\re}{\Re}

\newcommand{\bdm}{\begin{displaymath}}
\newcommand{\edm}{\end{displaymath}}
\newcommand{\beq}{\begin{equation}}
\newcommand{\eeq}{\end{equation}}
\newcommand{\beqa}{\begin{eqnarray}}
\newcommand{\eeqa}{\end{eqnarray}}
\newcommand{\nn}{\nonumber}

\newcommand{\E}{\mathcal{E}}
\newcommand{\gS}{\mathcal{S}}
\newcommand{\gL}{\mathcal{L}}

\newcommand{\ppo}{\bar{\gamma}} 
\newcommand{\po}{\tilde{\gamma}} 

\newcommand{\Graph}{G}

\newcommand{\Vertices}{\mathcal{V}}

\newcommand{\Bonds}{\mathcal{B}}

\newcommand{\Reals}{\mathbb{R}}

\newcommand{\CombinatorialGraph}{\mathcal{G}}

\begin{document}

\title[Finite pseudo orbit expansions for spectral quantities of quantum graphs]{Finite pseudo orbit expansions for spectral quantities of quantum graphs}

\author{R Band$^1$, J M Harrison$^{2}$ and C H Joyner$^1$}

\address{$^1$Department of Mathematics, University Walk, Clifton, Bristol BS8 1TW, UK}
\address{$^2$Department of Mathematics, Baylor University, Waco, TX 76798, USA}

\eads{\mailto{rami.band@bristol.ac.uk}, \mailto{jon\_harrison@baylor.edu}, \mailto{chris.joyner@bristol.ac.uk}}
\begin{abstract}
We investigate spectral quantities of quantum graphs by expanding them as sums over \emph{pseudo orbits}, sets of periodic orbits.  Only a finite collection of pseudo orbits which are irreducible and where the total number of bonds is less than or equal to the number of bonds of the graph appear, analogous to a cut off at half the Heisenberg time.  The calculation simplifies previous approaches to pseudo orbit expansions on graphs.  We formulate coefficients of the characteristic polynomial and derive a secular equation in terms of the irreducible pseudo orbits.  From the secular equation, whose roots provide the graph spectrum, the zeta function is derived using the argument principle.  The spectral zeta function enables quantities, such as the spectral determinant and vacuum energy, to be obtained directly as finite expansions over the set of short irreducible pseudo orbits.
\end{abstract}

\submitto{\JPA}
\maketitle

\section{Introduction}

\label{sec:intro}

Quantum graphs provide models in many disparate areas of mathematical physics.  They have been used to describe electrons in organic molecules, nanotechnology, wave guides, Anderson localization, quantum chaos and mesoscopic systems to name but a few, see \cite{p:K02} for a review.  One of the reasons why quantum graphs have been so successful is that they possess an elegant exact trace formula expansion for the spectrum in terms of
periodic orbits on the graph. Such formul{\ae} were first derived by Roth
\cite{Rot_incol83} and independently by Kottos and Smilansky \cite{KotSmi_prl97,p:KS99}, see \cite{p:BE09} for the current state of the art.
The prototype for quantum mechanical trace formul{\ae}
was provided by Gutzwiller who showed that in the semiclassical
limit the energy spectrum can be expressed as a sum involving an infinite number of classical periodic orbits
\cite{Gut_chaos90}.    For
most systems this trace formula suffers from convergence issues
that can only be resolved by limiting the contribution of long periodic
orbits and in the process valuable information is suppressed. One
way to recover this information is through a remarkable resummation
procedure which connects long and short pseudo orbits (collections
of periodic orbits) via a truncated quantum zeta function. The resulting
expression, discovered by Berry and Keating \cite{BerKea_jpa90,BerKea_prsl92,Kea_prsl92}
and derived using a different method by Bogomolny
\cite{Bog_non92},
is termed the Riemann-Siegel
lookalike formula due to the connection with its number theoretical
namesake. It allows one to approximately determine the quantum spectrum
up to some given energy, using only a finite number of periodic orbits.

For quantum graphs, there exists both an exact Riemann-Siegel
lookalike formula and an exact pseudo orbit expansion including only a finite number of pseudo orbits with at most as many bonds as the graph.   Kottos and Smilansky describe this without an explicit formula in \cite{p:KS99}.
Here we provide a basic argument to obtain a minimal pseudo orbit expansion.
This also sheds light on the origin of the cancelations seen in such pseudo orbit expansions.
We use the expansion to examine the coefficient of the characteristic polynomial and
to produce a simple expression for a secular equation whose roots provide the spectrum of the graph
$0\leqslant \lambda_{1}\leqslant \lambda_{2}\leqslant\dots$.  This enables the whole spectrum of the graph to be formulated in terms of a finite number of pseudo orbits which visit at most the same number of bonds as are in the graph.

To demonstrate the efficacy of this minimal pseudo orbit expansion we describe two applications to the spectral determinant and vacuum energy of graphs.
Formally the spectral determinant
is the product of the eigenvalues,
\begin{equation}
\gS(\lambda)={\prod_{n=1}^{\infty}}\phantom{|}^{\prime} (\lambda+\lambda_{n})\ ,\label{eq:formal spec det}
\end{equation}
where $\lambda$ is some spectral parameter and the prime indicates that zero modes are omitted where present. The spectral determinant on
graphs has been investigated by a number of authors \cite{p:BK99,p:CDT05,p:D00,p:D01,p:F06,p:HK10,p:HK11,p:HKT12,p:KMW03}.
In particular in \cite{AkkComDebMonTex_ap00} Akkermans et al.  derive a formula for the spectral determinant of the graph as a product over all periodic orbits.
After expanding the product they discuss cancelations between certain sets of periodic orbits using a diagrammatic approach.
While the pseudo orbit expansion that we describe in section \ref{sec:det} has not been written down previously, it can be obtained using the cancelation mechanisms from \cite{AkkComDebMonTex_ap00}, as is explained in the appendix.
The method that we employ in the body of the paper is more straightforward and the proof shows that the existence of such a pseudo orbit expansion does not depend on the structure of the scattering matrix, as implied in \cite{AkkComDebMonTex_ap00}.

A second application of the pseudo orbit expansion allows us to provide a trace formula over the short pseudo orbits
for the vacuum energy of the graph. Formally the vacuum energy is
\begin{equation}
\E_{c}=\frac{1}{2}\sum_{n=1}^{\infty}\phantom{|}^{\prime}\sqrt{\lambda_{n}}\ .\label{eq:formal Ec}
\end{equation}
 Vacuum energy is responsible for the Casimir effect, conventionally
seen as an attraction between uncharged parallel conducting plates
at short distances.  Analytic results for the vacuum energy
in most geometries are hard to obtain and consequently it
is interesting to consider mathematical models of vacuum energy where
the analysis can be taken further \cite{p:LMB01}.  In \cite{p:FKW07}
Fulling et. al. consider the vacuum energy of a star graph with Neumann conditions at the vertices and show that the Casimir effect is repulsive.  The vacuum energy for generic
graphs and vertex conditions was investigated in \cite{p:BHW09,p:HK10,p:HK11,p:HK12}.

To regularize both the spectral determinant and vacuum energy
we employ the spectral zeta function,
\begin{equation}
\zeta(s,\lambda)={\sum_{n=1}^{\infty}}\phantom{|}^{\prime}(\lambda+\lambda_{n})^{-s}\ ,\label{eq:defn spec zeta}
\end{equation}
 which is convergent for $\re s>1$.
Making an analytic continuation to the complex $s$ plane the spectral
determinant is defined to be,
\begin{equation}
\gS(\lambda):=\exp(-\zeta'(0,\lambda))\ .\label{eq:spec det defn}
\end{equation}
 Similarly the regularized vacuum energy is obtain from the $\zeta(-1/2,0)$,
see section \ref{sec:vac}.

The article is laid out as follows. Section \ref{sec:intro_quantum_graphs} introduces quantum graphs and the scattering approach used to define a secular function.  We then provide, in section \ref{sec:expansion_for_secular_function}, a simple derivation for the pseudo orbit expansion of the secular function. The following sections are devoted to applications of this expansion. Section \ref{sec:coefficients_of_char_poly} examines the variance of the characteristic polynomial coefficients, section \ref{sec:zeta} develops a regularized expression for the zeta function and section \ref{sec:trace} uses this expression to obtain pseudo orbit expansions for the spectral determinant and vacuum energy. After summarizing the results we include an appendix which describes alternative cancelation mechanisms which provide pseudo orbits expansions.

\section{Quantum graphs}\label{sec:intro_quantum_graphs}
Let $\CombinatorialGraph=(\Vertices,\Bonds)$ be a graph with the finite sets of
vertices $\Vertices$ and bonds (or edges) $\Bonds$. Each bond, $b\in\Bonds$,
connects a pair of vertices, $u,v\in\Vertices$, see figure \ref{fig:fcs}.
We use the notation $b=\left(u,v\right)$ to denote a directed bond where $o\left(b\right)=u$ and
$t\left(b\right)=v$, with $o$ and $t$ standing for the origin and
the terminus of $b$.

We denote by $\hat{b}=\left(v,u\right)$
the reverse of the directed bond $b$.  For every bond $b\in \Bonds$ we insist that
$\hat{b}$ is also in $\Bonds$. The size of the set $\Bonds$ is
then denoted by $2B:=\left|\Bonds\right|$. The set of all bonds whose
origin is the vertex $v$ is denoted $\Bonds_{v}$ and its size is
called the degree of the vertex, $d_{v}:=\left|\Bonds_{v}\right|$.

The combinatorial graph $\CombinatorialGraph$ becomes a metric graph $\Graph$ when we identify the bonds of $\CombinatorialGraph$
with one dimensional intervals $[0,l_b]$ with a positive finite length $l_{b}$
($l_{\hat{b}}=l_{b}$). We can then assign to each bond
$b$ a coordinate, $x_{b}$, which measures the distance
along the bond starting from $o(b)$. The reverse directed
bond, $\hat{b}$, posses the coordinate, $x_{\hat{b}}$, and the two
coordinates are related by $x_{b}+x_{\hat{b}}=l_{b}$. We denote
a coordinate by $x$, when its precise nature is unimportant.
$2\mathcal{L}=\sum_{b\in \Bonds} L_b$ is twice the total length of $\Graph$.

\begin{figure}[!ht]
  \begin{center}
  \setlength{\unitlength}{1.3cm}
    \begin{picture}(9,3)
    \put(0,0.5){\includegraphics[width=5.2cm]{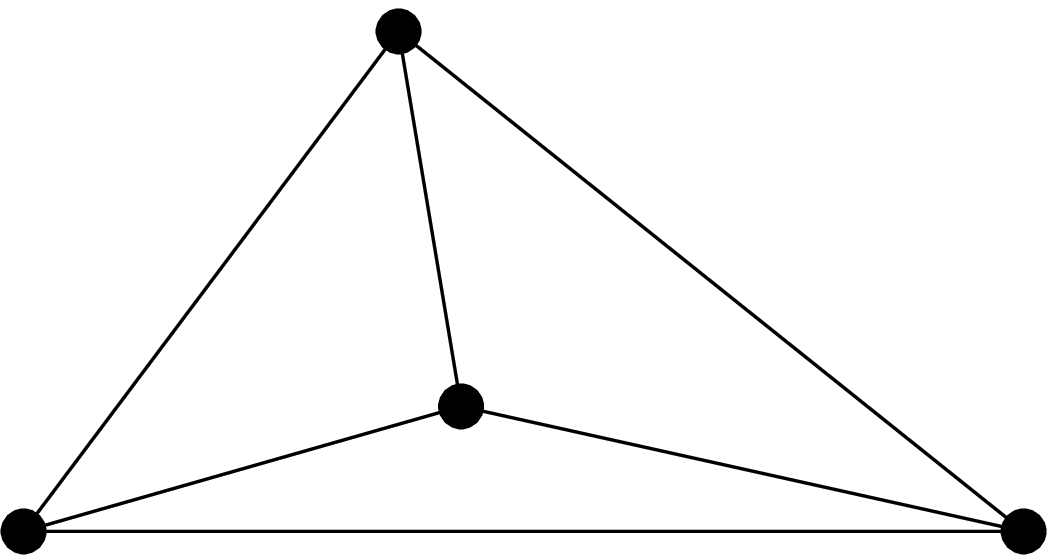}}
    \put(5.5,0){\includegraphics[width=4.55cm]{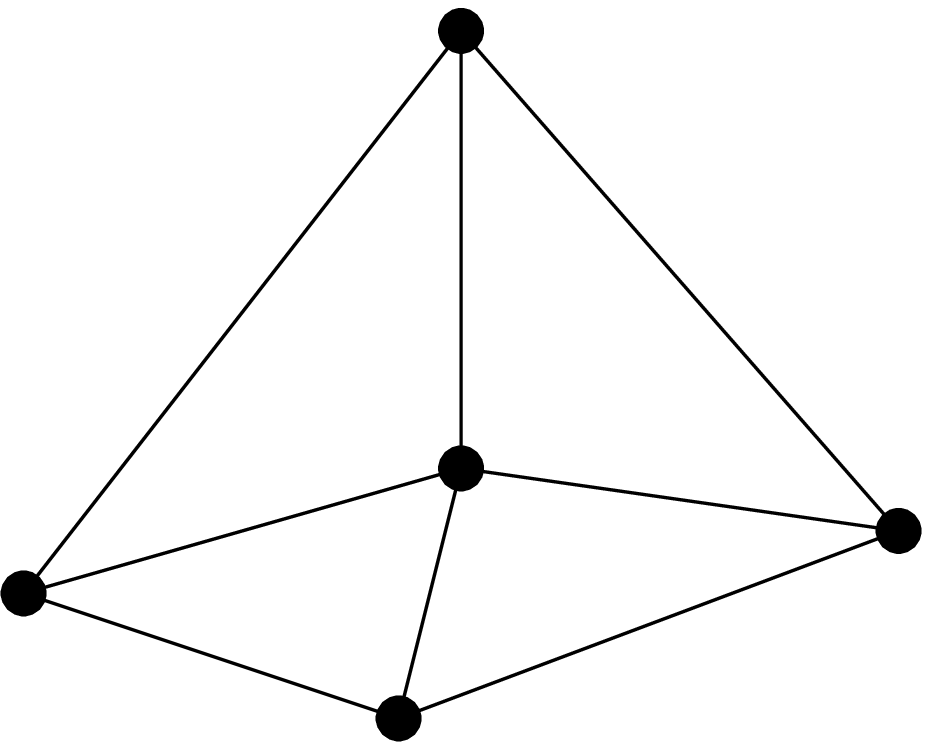}}
    \put(0,2.6){(a)}
    \put(5.5,2.6){(b)}
    \end{picture}
  \caption{\it (a) A fully connected graph with four vertices and six bonds and (b) a graph with five vertices and eight bonds.}
  \label{fig:fcs}
  \end{center}
\end{figure}

The quantum graphs we consider are metric graphs equipped with a self-adjoint
differential operator $\mathcal{H}$, the Hamiltonian.
Here we are particularly interested in the negative Laplace operator,
\begin{equation}
\mathcal{H}\ :\ f(x)\mapsto-\frac{d^{2}f}{dx^{2}} \ , \label{E:lap_op}
\end{equation}
 or the more general Schr\"odinger operator,
\begin{equation}
\mathcal{H}\ :\ f(x)\mapsto-\frac{d^{2}f}{dx^{2}}+V(x)f(x) \ ,\label{E:schrod}
\end{equation}
 where $V(x)$ is a \emph{potential}, which we assume to be bounded
and piecewise continuous. Note that the value of a function or the
second derivative of a function at a point on the bond is well-defined,
thus it is not important which coordinate, $x_{b}$ or $x_{\hat{b}}$
is used. This is in contrast to the first derivative which changes
sign according to the direction of the chosen coordinate.

To complete the definition of the operator we need to specify its
domain.  We denote by $H^{2}(\Graph)$
the space
\begin{equation}
H^{2}(\Graph):=\bigoplus_{b=1}^B H^{2}([0,l_b])\ ,
\end{equation}
 which consists of the functions $f$ on $\Graph$ that on each bond $b$
belongs to the Sobolev space $H^{2}([0,l_b])$. The restriction of $f$ to
the bond $b$ is denoted by $f_{b}$, where $f_b(x_b)=f_{\hat{b}}(l_b-x_{{b}})$. The norm in the space $H^{2}(\Graph)$
is
\begin{equation}
\|f\|_{H^{2}(\Graph)}:=\sum\limits _{b=1}^B\|f_{b}\|_{H^{2}([0,l_b])}^{2} \ .
\end{equation}
Note that in the definition of $H^{2}(\Graph)$ the smoothness
is enforced along bonds only, without any matching conditions at the
vertices at all. However, the standard Sobolev trace theorem (see e.g.
\cite{EdmundsEvans_spectral87}) implies that each function $f_{b}$
and its first derivative have well-defined values at the endpoints
of the bond $b$.

All conditions matching functions at the vertices that lead to the operator (\ref{E:lap_op}) being self-adjoint
have been classified in \cite{Har_jpa00,KosSch_jpa99,Kuc_wav04}.
It can be shown that under these conditions the operator $\mathcal{H}$ is bounded
from below \cite{Kuc_wav04}. In addition, since we only consider compact
graphs, the spectrum is real, discrete and with no accumulation points.
Thus we can number the eigenvalues in the ascending order and denote
them with $\left\{ \lambda_{n}\right\} _{n=1}^{\infty}$. The $n$'th
eigenvalue obeys the following equation
\begin{equation}
-\frac{d^{2}f_{n}}{dx^{2}}+V(x)f_{n}(x)=\lambda_{n}f_{n}(x) \ ,\label{eq:eig_eq}
\end{equation}
where $f_{n}$ is the corresponding eigenfunction. We also
use $k_{n}$, such that $\lambda_{n}=k_{n}^{2}$, and say that $\left\{ k_{n}\right\} _{n=1}^{\infty}$
is the $k$-spectrum of the graph.

In this paper we will primarily consider
graphs where the domain of the operator (\ref{E:schrod}) consists
of the functions $f\in H^{2}(\Graph)$ which obey Neumann conditions
at all vertices of the graph.
The function $f$ is continuous at the vertex $v$, i.e.
$f_{b_{1}}(0)=f_{b_{2}}(0)$ for all bonds $b_{1},\, b_{2}\in\Bonds_{v}$.
The derivatives of $f$ at the vertex $v$ satisfy,
\begin{equation}
\sum_{b\in\Bonds_{v}}\frac{df_b}{dx_{b}}(0)=0 \ .\label{eq:delta_deriv}
\end{equation}
We will assume that $\Graph$ contains no loops (bonds $b$ with $o(b)=t(b)$) and no pairs of vertices connected by multiple bonds.  However, as a loop or bond can be subdivided by introducing a vertex of degree two with Neumann boundary conditions without changing the spectrum, this assumption does not reduce the generality of the results.

In addition, we will assume that the graph's spectrum
is non-negative, $\lambda_{1}\geq0$. This assumption holds in the case
of graphs with Neumann conditions with no potential, where $\lambda_{1}=0$.
The condition can be relaxed by accepting some additional technicalities, see \cite{p:KM04}.
The next section describes an approach to obtain the spectrum of such
a graph.

\subsection{The scattering approach to the graph spectrum}\label{sec:scat_mat}
Let $\Graph$ be a graph with Neumann vertex conditions and no potential. The eigenvalue equation,
\begin{equation}
-\frac{d^{2}f}{dx^{2}}=k^{2}f(x) \ ,\label{eq:eig_eq_V0}
\end{equation}
on each bond has a solution that can be written as a linear combination of
two complex exponentials if $k\neq0$. We write
\begin{equation}
f_{b}(x_{b})=a_{b}^{\textrm{in}}\ue^{-\rmi kx_{b}}+a_{b}^{\textrm{out}}\ue^{\rmi kx_{b}} \ ,\label{eq:sol_on_edge_a}
\end{equation}
 for the solution on the bond $b$.
%
The same function can be expressed using the coordinate $x_{\hat{b}}$
as
\begin{equation}
f_{\hat{b}}(x_{\hat{b}})=a_{\hat{b}}^{\textrm{in}}\ue^{-\rmi kx_{\hat{b}}}+a_{\hat{b}}^{\textrm{out}}\ue^{\rmi kx_{\hat{b}}} \ .\label{eq:sol_on_edge_other}
\end{equation}
 Since these two expressions should define the same function and since
the two coordinates are related, through the identity $x_{b}+x_{\hat{b}}=l_{b}$,
we arrive at the following relations
\begin{equation}
a_{b}^{\textrm{in}}=\rme^{\rmi kl_{b}}a_{\hat{b}}^{\textrm{out}}\qquad \textrm{and} \qquad a_{\hat{b}}^{\textrm{in}}=\rme^{\rmi kl_{b}}a_{b}^{\textrm{out}} \ .\label{eq:a_connection}
\end{equation}

Fixing a vertex $v$ and using the matching conditions to relate solutions
$f_{b}$ for all bonds $b$ with $o\left(b\right)=v$
one arrives at
\begin{equation}
\vec{a}_{v}^{\,\textrm{out}}=\sigma^{\left(v\right)}(k)\vec{a}_{v}^{\,\textrm{in}},\label{eq:vert_scat}
\end{equation}
 where $\vec{a}_{v}^{\,\textrm{out}}$ and $\vec{a}_{v}^{\,\textrm{in}}$ are vectors of
the outgoing and incoming coefficients at $v$ and $\sigma^{(v)}(k)$ is a
$d_{v}\times d_{v}$ unitary matrix, $d_{v}$ being the degree of
the vertex $v$. The matrix $\sigma^{(v)}(k)$ is called the vertex-scattering
matrix. It is $k$ independent for the Neumann vertex conditions and
its entries were calculated in \cite{p:KS99}:
\begin{equation}
\sigma_{b,b'}^{\left(v\right)}=\frac{2}{d_{v}}-\delta_{b,b'} \ .
\end{equation}

Collecting all coefficients $a_{b}^{\textrm{in}}$ on the whole graph into a vector $\vec{a}$
of size $2B$ such that the first $B$ entries
correspond to bonds which are the inverses of the last $B$ entries.
We can then define the matrix $J$ acting on $\vec{a}$ by requiring
that it exchanges $a_{b}^{\textrm{in}}$ and $a_{\hat{b}}^{\textrm{in}}$ for all $b$ such that,
\begin{equation}
J=\left(\begin{array}{cc}
0 & \UI\\
\UI & 0
\end{array}\right).
\end{equation}
Then, collecting equations (\ref{eq:vert_scat}) for all vertices into one system and
using (\ref{eq:a_connection}) we have,
\begin{equation}
J\rme^{-\rmi kL}\vec{a}=\Sigma(k)\vec{a} \ ,
\end{equation}
 where $L=\textrm{diag}\{l_1,\dots,l_B,l_1,\dots,l_B\}$ is a diagonal matrix of bond lengths and $\Sigma\left(k\right)$
is block-diagonalizable with individual $\sigma^{(v)}\left(k\right)$
as blocks. This can be rewritten as (note that $J^{-1}=J$),
\begin{equation}
\vec{a}=\rme^{\rmi kL}J\Sigma(k)\vec{a}\ .\label{eq:bond_scat}
\end{equation}
The unitary matrix $S(k):=J\Sigma(k)$ is called the \emph{bond-scattering}
matrix. The unitary matrix $U\left(k\right):=\rme^{\rmi kL}S(k)$
is the \emph{quantum evolution operator}. Using this notation we
write (\ref{eq:bond_scat}) as
\begin{equation}
\left(\UI-U\left(k\right)\right)\vec{a}=\vec{0} \  .\label{eq:secular_cond_tmp}
\end{equation}

A vector $\vec{a}$ satisfying equation (\ref{eq:secular_cond_tmp}) defines
an eigenfunction on the graph with the eigenvalue $E=k^{2}$. Hence we
arrive at the conclusion that the non zero eigenvalues of the graph
are the solutions of
\begin{equation}
\det\left(\UI-U\left(k\right)\right)=0 \ .\label{eq:secular_cond}
\end{equation}
 This is called the \emph{secular equation}. It is the major
tool for investigating various spectral properties of quantum graphs.
The interested reader can find a more elaborate description of the
derivation above in \cite{GnuSmi_aip06}. We end by noting that a
similar equation may be derived for the case of a potential which
differs from zero \cite{RueSmi_unpub}.

\section{A pseudo orbits expansion for the secular function}\label{sec:expansion_for_secular_function}
We define the \emph{secular function}
\begin{equation}
F\left(k\right):=\det\left(U\left(k\right)\right)^{-\frac{1}{2}}\det\left(\UI-U\left(k\right)\right).
\end{equation}
This is the LHS of the secular equation (\ref{eq:secular_cond}),
multiplied by $\det\left(U\left(k\right)\right)^{-\frac{1}{2}}$ to
turn it into a real function for all $k\in\Reals$, as $U\left( k \right)$ is unitary. We recall that all positive values
of the graph's $k$-spectrum are given by the zeros of $F\left(k\right)$.

The characteristic polynomial of $U\left(k\right)$ is,
\begin{equation}
F_{\xi}\left(k\right):=\det\left(\xi\UI-U\left(k\right)\right)=\sum_{n=0}^{2B}a_{n}\xi^{2B-n} \ ,\label{eq:characteristic_polynomial}
\end{equation}
where the $k$-dependence of the coefficients is omitted for brevity.
One can express the secular function in terms of the coefficients of the characteristic
polynomial
\begin{eqnarray}
F\left(k\right) & = & \det\left(U\left(k\right)\right)^{-\frac{1}{2}}F_{1}\left(k\right)\nonumber \\
 & = & a_{2B}^{-\frac{1}{2}}\sum_{n=0}^{2B}a_{n} \ .\label{eq:secular_fucnction1}
\end{eqnarray}
The main result of this section is a theorem which expands the secular
function as a sum over pseudo orbits on the graph. Before stating
the theorem, we describe orbits on graphs.

A \emph{periodic orbit} on $\Graph$ is a closed path that starts and ends at the same vertex.
We denote such a trajectory by the corresponding
ordered set of its bonds, $\gamma:=\left(b_{1},b_{2},\ldots b_{n}\right)$,
such that $\forall i\,,\,\, t\left(b_{i}\right)=o\left(b_{i+1}\right)$
and $t\left(b_{n}\right)=o\left(b_{1}\right)$.
Cyclic permutations of the bonds of $\gamma$ are the same periodic orbit,
for example, $\gamma=\left(b_{2},\ldots b_{n},b_{1}\right)$.

A \emph{primitive periodic orbit} is a periodic orbit which
cannot be written as a repetition of a shorter periodic
orbit. We can thus define the \emph{repetition number} of an orbit $r_{\gamma}$ as the
number of repetitions of the primitive periodic orbit which it
contains. For example, if $\gamma_{0}=\left(b_{1},b_{2},b_{3}\right)$
is a primitive periodic orbit, then $\gamma_{1}=\left(b_{1},b_{2},b_{3},b_{1},b_{2},b_{3}\right)$
is a periodic orbit with repetition number $r_{\gamma_1}=2$.
Note that the repetition number of a primitive periodic orbit is always one.
For a periodic orbit $\gamma$ we use $B_\gamma:=n$ to denote the number of bonds in $\gamma$, the \emph{topological length} of $\gamma$.  Similarly the \emph{metric length} (or length) of $\gamma$ is denoted $l_\gamma:=\sum_{b_j\in \gamma} l_{b_j}$.  The product of scattering amplitudes around $\gamma$ is;
\begin{equation}\label{eq:defn A_gamma}
    A_\gamma=S_{b_2 b_1} S_{b_3 b_2} \dots S_{b_{n} b_{n-1}}  S_{b_{1} b_n}\ ,
\end{equation}
where we have suppressed the $k$ dependence if present.

We now define a more complex structure which we call a pseudo orbit
(also called a composite orbit in the literature).
A \emph{pseudo orbit} is a collection of periodic orbits, $\po=\left\{ \gamma_{1},\gamma_{2},\ldots\gamma_{M}\right\}$.
An \emph{irreducible pseudo orbit} $\ppo$ is a pseudo orbit which contains
no directed bond more than once. So every periodic orbit $\gamma_j$ in $\ppo$ has at most one copy of each directed bond and every pair, $\gamma_i,\gamma_j$, has no directed bond in common. In particular this implies that an
irreducible pseudo orbit consists only of primitive periodic orbits
and the primitive periodic orbits are distinct.

We will use the following
notation for properties of $\po$:
\begin{eqnarray}
  m_{\po}&:= M &\textrm{ the number of periodic orbits in } \po . \\
  B_{\po}&:= \sum_{\gamma_j \in \po} B_{\gamma_j} &\textrm{ the total topological length of } \po .\\
  l_{\po}&:= \sum_{\gamma_j \in \po} l_{\gamma_j} &\textrm{ the total metric length of } \po .\\
  A_{\po}&:=  \prod_{\gamma_j \in \po} A_{\gamma_j} \ .
\end{eqnarray}
The null pseudo orbit $\overline{0}$ is the pseudo orbit which
contains no orbits. We treat it as an irreducible pseudo orbit
with the following properties: $m_{\overline{0}}=B_{\overline{0}}=l_{\overline{0}}=0$ and $A_{\overline{0}}=1$.

\begin{theorem}\label{thm:a_n_expansion}

The coefficients of the characteristic polynomial
$F_{\xi}\left(k\right)$ are given by
\begin{equation}
a_{n}=\sum_{\ppo |\, B_{\ppo}=n}\left(-1\right)^{m_{\ppo}}A_{\ppo}\left(k\right)\exp\left(\mathrm{i} kl_{\ppo}\right).\label{eq:a_n_expansion}
\end{equation}

\end{theorem}

\begin{proof}

We first rewrite the characteristic polynomial using the permutation expansion
of the determinant
\begin{equation}
\det\left(\xi\UI-U\left(k\right)\right)=\sum_{\rho \in S_{2B}}\textrm{sgn}\left(\rho \right)\prod_{b=1}^{2B}\left[\xi\UI-U\left(k\right)\right]_{\rho \left(b\right),b},\label{eq:permutation_expansion}
\end{equation}
where the indices of the matrix, $\rho\left(b\right),b$, label directed bonds of $\CombinatorialGraph$. We may now use the unique presentation
of the permutations as a product of disjoint cycles
\begin{equation}
\rho=\left(b_{1},b_{2},\ldots b_{n_{1}}\right)\left(b_{n_{1}+1},\ldots, b_{n_{1}+n_{2}}\right)\ldots\left(b_{\sum_{j=1}^{m_{\rho}-1}n_{j}+1}, \ldots b{}_{\sum_{j=1}^{m_{\rho}}n_{j}}\right) \ ,\label{eq:cycle_representation}
\end{equation}
for a permutation $\rho$ with $m_{\rho}$ cycles where the $j$'th cycle has length $n_j$. The notation
above means $\rho$ acts on the $2B$ directed bonds such that
$\rho\left(b_{1}\right)=b_{2}$ etc.
 The representation
(\ref{eq:cycle_representation}) is unique up to cyclic permutations
of each cycle and no directed bond appears more than once.

Recall that $\left[U\left(k\right)\right]_{\rho\left(b\right),b}\neq0$
only if the directed bond $\rho\left(b\right)$ is connected to the directed
bond $b$, i.e., $o\left(\rho\left(b\right)\right)=t\left(b\right)$.\footnote[1]{Provided that all elements of the vertex scattering matrices are non-zero.  This would not hold, for example, in the case of equi-transmitting matrices \cite{p:HSW07} where back-scattering is prohibited.}
From which we conclude that a permutation $\rho\neq \textrm{I}$ can have a non
zero contribution to the sum (\ref{eq:permutation_expansion}) only
if $o\left(\rho\left(b\right)\right)=t\left(b\right)$ for all $b$ in $\rho$.
This allows us to interpret (\ref{eq:cycle_representation})
as an irreducible pseudo orbit on the graph. Each
periodic orbit in the pseudo orbit corresponds to a cycle of $\rho$.
The pseudo orbit is irreducible since no directed bond appears
more than once in (\ref{eq:cycle_representation}).  The fixed points of $\rho$, the bonds $b$ not in $\rho$ so $b=\rho\left(b\right)$,
each produce a multiplicative factor of $\xi$ in
the contribution of $\rho$ in (\ref{eq:permutation_expansion}).

The coefficient $a_{n}$ multiplies $\xi^{2B-n}$ and is therefore
obtained as a sum of all permutations $\rho\in S_{2B}$ with $2B-n$
fixed points. Each such permutation corresponds to an irreducible
pseudo orbit, $\ppo$, whose bond length is $B_{\ppo}=n$. Using the
notation of (\ref{eq:cycle_representation}), the pseudo orbit $\rho$ consists
of $m_{\rho}$ orbits and $n=\sum_{j=1}^{m_{\rho}}n_{j}$. The
contribution of this permutation (irreducible pseudo orbit) to (\ref{eq:permutation_expansion})
is given by
\[
\fl \textrm{sgn}\left(\rho\right)\xi^{2B-n}\left(-1\right)^{\sum_{j=1}^{m_{\rho}}n_{j}}\left(U_{b_{2}b_{1}}U_{b_{3}b_{2}}\ldots U_{b_{1}b_{n_{1}}}\right)\left(U_{b_{n_{1}+2}b_{n_{1}+1}}\ldots U_{b_{n_{1}+1}b_{n_{1}+n_{2}}}\right)\ldots,
\]
where we omitted the $k$ dependence from $U$ for clarity.
If we recall 
that $U_{b_{2}b_{1}}=\exp\left(\ui kl_{b_{2}}\right)S_{b_{2}b_{1}}$,
we see that the contribution above (up to a sign) is $\xi^{2B-n}\exp\left(\ui kl_{\ppo}\right)A_{\ppo}$.
The result (\ref{eq:a_n_expansion}) now follows since
\begin{eqnarray*}
\fl\textrm{sgn}\left(\rho\right)\left(-1\right)^{\sum_{j=1}^{m_{\rho}}n_{j}} & = & \left(-1\right)^{\sum_{j=1}^{m_{\rho}}\left(n_{j}+1\right)}\left(-1\right)^{\sum_{j=1}^{m_{\rho}}n_{j}}=\left(-1\right)^{m_{\rho}}=\left(-1\right)^{m_{\rho}}.
\end{eqnarray*}

\end{proof}

We now employ the unitarity of $U\left(k\right)$ to get a Riemann-Siegel lookalike formula which connects pairs of coefficients of the characteristic polynomial \cite{p:KS99}. The characteristic polynomial can be rewritten
\begin{eqnarray}
F_{\xi}\left(k\right) & = & \det\left(-\xi U\left(k\right)\right)\det\left(\xi^{-1}\UI-U^{\dagger}\left(k\right)\right) \nn \\
 & = & \det\left(-\xi U\left(k\right)\right)\sum_{n=0}^{2B}a_{n}^{*}\xi^{n-2B} \nn \\
 & = & \det\left(U\left(k\right)\right)\sum_{n=0}^{2B}a_{n}^{*}\xi^{n} \ .
\end{eqnarray}
 Comparing this to $F_{\xi}\left(k\right)=\sum_{n=0}^{2B}a_{n}\xi^{2B-n}$ and using
$a_{2B}=\det\left(U\left(k\right)\right)$, we deduce the following
symmetry relation of the characteristic polynomial coefficients
\begin{equation}
a_{n}=a_{2B}a_{2B-n}^{*} \ .\label{eq:symmetry_of_char_pol_coefficients}
\end{equation}
This symmetric expression allows us to express the secular function in terms of shorter pseudo orbits, whose total number of bonds is less than or equal to the number of bonds of the graph.


\begin{theorem}\label{thm:secular function}
The secular function has the following expansion in terms of irreducible pseudo orbits,
\begin{equation}\label{eq:Riemann-Siegel_expression}
F\left(k\right) = 2\sum_{\ppo | B_{\ppo} \leq B}\left(-1\right)^{m_{\ppo}}A_{\ppo}\cos\left(k\left(l_{\ppo}-\mathcal{L}\right)-\theta\left(k\right)\right)H\left(B-B_{\ppo}\right) \ ,
\end{equation}
where $\det S=\exp\left(\mathrm{i}2\theta\left(k\right)\right)$, and $H$ is the Heaviside function, see (\ref{eq:heavyside}).
\end{theorem}

\begin{proof}

Plugging the symmetry relation (\ref{eq:symmetry_of_char_pol_coefficients}) in the secular function expression (\ref{eq:secular_fucnction1}) we obtain,
\begin{eqnarray}\label{eq:Riemann Siegel}
F\left(k\right) & = & \det\left(U\left(k\right)\right)^{-\frac{1}{2}}\sum_{n=0}^{2B}a_{n}\nonumber \\
 & = & \sum_{n=0}^{B-1}\left(a_{2B}^{-\frac{1}{2}}a_{n}+a_{2B}^{\frac{1}{2}}a_{n}^{*}\right)+a_{2B}^{-\frac{1}{2}}a_{B}\nonumber \\
 & =\Re & \left\{ 2\sum_{n=0}^{B-1}a_{2B}^{-\frac{1}{2}}a_{n}+a_{2B}^{-\frac{1}{2}}a_{B}\right\} \ ,
\end{eqnarray}
where $a_{2B}^{-1/2}a_B$ is a real number which can be seen from the symmetry relation as $|a_{2B}|=1$.
Combining this with theorem \ref{thm:a_n_expansion} we see
\begin{eqnarray}
 F\left(k\right)
  &= & \Re\left\{ 2\left(\det S\right)^{-\frac{1}{2}}\sum_{n=0}^{B-1}\sum_{\ppo| \, B_{\ppo}=n}\left(-1\right)^{m_{\ppo}}A_{\ppo}\exp\left(\ui k\left(l_{\ppo}-\mathcal{L}\right)\right) \right.
\nn \\
 &&\left. +\left(\det S\right)^{-\frac{1}{2}}\sum_{\ppo | B_{\ppo}=B}\left(-1\right)^{m_{\ppo}}A_{\ppo}\exp\left(\ui k\left(l_{\ppo}-\mathcal{L}\right)\right)\right\} .\label{eq:secular_function3}
\end{eqnarray}
The unitarity of $S$ allows to denote $\det S=\exp\left(\ui 2\theta\left(k\right)\right)$
and to rewrite the above as
\begin{eqnarray}
 F\left(k\right)  = & 2\sum_{n=0}^{B-1}\sum_{\ppo |\, B_{\ppo}=n}\left(-1\right)^{m_{\ppo}}A_{\ppo}\cos\left(k\left(l_{\ppo}-\mathcal{L}\right)-\theta\left(k\right)\right)\nn \\ &+\sum_{\ppo |\, B_{\ppo}=B}\left(-1\right)^{m_{\ppo}}A_{\ppo}\cos\left(k\left(l_{\ppo}-\mathcal{L}\right)-\theta\left(k\right)\right).
\end{eqnarray}
This expression is made more compact by using the Heaviside
step function
\begin{equation}\label{eq:heavyside}
    H\left(x\right)=\left\{ \begin{array}{cc}
0 & x<0\\
\frac{1}{2} & x=0\\
1 & x>0
\end{array}\right. \ .
\end{equation}


\end{proof}

Finally, if we make the assumption that our graph has Neumann conditions at the vertices then either $\theta=0$ or $\theta=\frac{\pi}{2}$ with no $k$-dependence.
This results in an expansion in terms of either cosine or sine functions.
In the remainder of the article we will assume that the graph has Neumann vertex conditions.

We also note that the point of truncation $B_{\ppo}\leq B$ in (\ref{eq:Riemann-Siegel_expression}) is analogous to a cutoff at half of the Heisenberg time.  The Heisenberg time $T_H$ is $2\pi$ multiplied by the mean spectral density, which for a graph is $\mathcal{L}/\pi$, hence $T_H=2\mathcal{L}$.  If $l_{\textrm{av}}$ is the average length of a bond on our graph, and we assume that the bond lengths are all close to $l_{\textrm{av}}$, then $l_{\ppo}\approx l_{\textrm{av}} B_{\ppo}$ and $T_H/2= l_{\textrm{av}} B$.  Thus the cutoff $B_{\ppo}\leq B$ resembles $l_{\ppo}\leq T_H/2$ used in the semiclassical expansions \cite{BerKea_jpa90,BerKea_prsl92,Kea_prsl92}.  This in turn originated with the cutoff introduced in the functional equation of the Riemann zeta function, known as the Riemann-Siegel formula.
However, the argument leading to the cutoff in (\ref{eq:Riemann-Siegel_expression}) is due to the unitarity of $U(k)$.  A similar approach applies to the dynamical zeta function \cite{Bog_chaos92,Bog_non92}.


\section{Variance of coefficients of the characteristic polynomial}\label{sec:coefficients_of_char_poly}
As a first application of these results we are now in a position to write a simple expansion for the variance of coefficients of the characteristic polynomial (\ref{eq:a_n_expansion}).
The mean values of the coefficients (other than $a_{0}=1$) are zero as the average over $k$ of $ \ue^{\ui k l_{\ppo}}$ is necessarily zero for pseudo orbits of topological length $n \geq 1$.
The second moment of the coefficients of the characteristic polynomial, the \emph{variance}, was investigated numerically in \cite{p:KS99} where it is seen that they may deviate from the random matrix predictions even for graphs whose spectral-statistics match those of appropriate random matrix distributions. Tanner \cite{p:Tan02} considered the variance of the coefficients for a related graph model which generates a unitary stochastic matrix ensemble \cite{p:Tan01} over which the average is taken.

From theorem \ref{thm:a_n_expansion} the variance of the coefficients of the characteristic polynomial can be expressed as a sum over pairs of irreducible pseudo orbits $\ppo, \ppo '$,
\begin{eqnarray}
  \langle |a_n|^2 \rangle_k &=
  \sum_{\ppo, \ppo' | B_{\ppo} = B_{\ppo '} =n} (-1)^{m_{\ppo}+m_{\ppo '}} A_{\ppo} A_{\ppo '}
  \lim_{K\to \infty} \frac{1}{K} \int_0^K \ue^{\ui k (l_{\ppo}-l_{\ppo'})} \ud k  \nn \\
   &=  \sum_{\ppo, \ppo' | B_{\ppo} = B_{\ppo '} =n} (-1)^{m_{\ppo}+m_{\ppo '}} A_{\ppo} A_{\ppo '} \,
   \delta_{l_{\ppo},l_{\ppo '}} \ . \label{eq:a_n variance}
\end{eqnarray}
So a pair of irreducible pseudo orbits of topological length $n$ only contributes to the sum when the total metric lengths of the pseudo orbits are equal.
When the set of bond lengths is incommensurate this implies that the two irreducible pseudo orbits visit a given pair of directed bonds $\{b,\hat{b}\}$ the same number of times. Hence the variance is independent of the bond lengths, provided that they are incommensurate.
The formulation of the variance (\ref{eq:a_n variance}) also holds for vertex scattering matrices that do not correspond to Neumann matching conditions provided that the scattering matrices are independent of $k$.

\subsection{Diagonal approximation}
In the case of incommensurate bond lengths the first order contribution to the variance will come from pairing an irreducible pseudo orbit $\ppo$ with itself, $\ppo '=\ppo$, or with any other orbit $\ppo'$ obtained from $\ppo$ by reversing the direction of a subset of orbits in the pseudo orbit, as in the diagonal approximation of Berry \cite{p:B85}.  As $\ppo$ contains $m_{\ppo}$ distinct orbits there are typically $2^{m_{\ppo}}$ different pseudo orbits that can be obtained from $\ppo$ by reversing a subset of the orbits in $\ppo$.  Making this diagonal approximation we obtain
\begin{equation}\label{eq:a_n variance diag}
    \langle |a_n|^2 \rangle_k \approx \sum_{\ppo | B_{\ppo} =n} 2^{m_{\ppo}} A_{\ppo }^2 \ .
\end{equation}
Note that while making the diagonal approximation, we did not take into account the existence of self retracing orbits. These orbits are equal to their reversed orbit. Pseudo orbits which contain self retracing orbits should therefore have a smaller prefactor in (\ref{eq:a_n variance diag}). However, the proportion of self retracing orbits among all periodic orbits decreases rapidly with their length, which justifies the approximation.

As $A_{\ppo}^2\leq 1$ this suggests that the dominant contribution in the diagonal approximation will come from pseudo orbits comprised of many distinct periodic orbits where $m_{\ppo}$ is maximal.
Figures \ref{fig:variance} and \ref{fig:variance2} show the variance of the coefficients of the characteristic polynomial for the two graphs shown in figure \ref{fig:fcs} computed from the pseudo orbit expansion (\ref{eq:a_n variance}), the diagonal approximation (\ref{eq:a_n variance diag}) and by directly averaging coefficients of $F_{\xi}\left(k\right)$ over a range of approximately $2000$ mean level spacings.
The diagonal approximation works well for these examples.  One can also see that where the diagonal approximation fails, around $n/2B=1/2$, the diagonal approximation over estimates the actual value which suggests that for these small graphs the over counting of the self-retracing orbits is significant.

%
%


\begin{figure}[!ht]
  \begin{center}
  \setlength{\unitlength}{1cm}
    \begin{picture}(12,7.5)
    \put(0.2,0){\includegraphics[width=10cm]{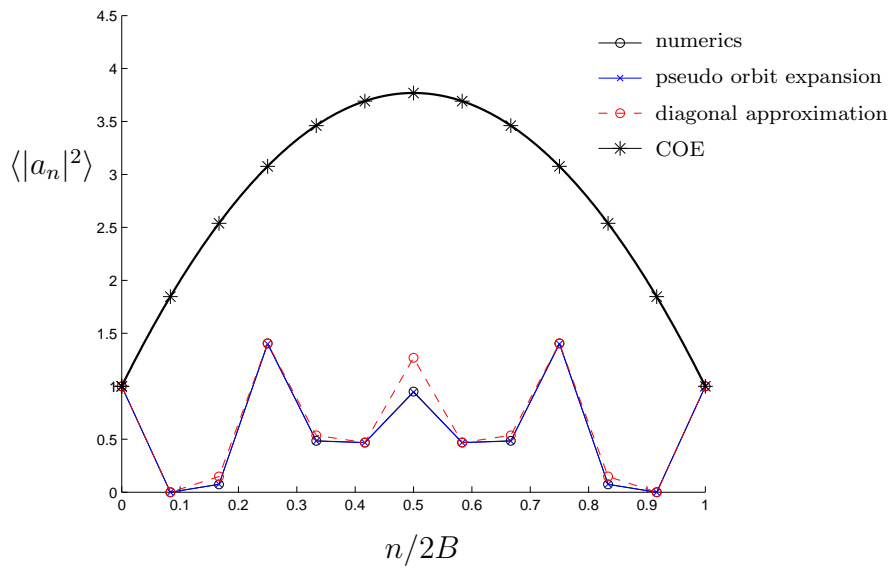}}
    \put(5,0){$n/2B$}
    \put(0,5.1){$\langle |a_n|^2 \rangle$}
    \put(8.6,6.79){\scriptsize numerics}
    \put(8.6,6.31){\scriptsize pseudo orbit expansion}
    \put(8.6,5.83){\scriptsize diagonal approximation}
    \put(8.6,5.35){\scriptsize COE}
    \end{picture}
  \caption{\it Variance of coefficients of the characteristic polynomial of a fully connected graph with four vertices (shown in figure \ref{fig:fcs}(a)) with random bond lengths chosen uniformly in $[0,1]$.  Numerical results were averaged over $200\,000$ values of $k$ over a range of approximately $2000$ mean level spacings and are compared with the pseudo orbit expansion, diagonal approximation and random matrix theory.}
  \label{fig:variance}
  \end{center}
\end{figure}

\begin{figure}[!ht]
  \begin{center}
  \setlength{\unitlength}{1cm}
    \begin{picture}(12,7.5)
    \put(0.2,0){\includegraphics[width=10cm]{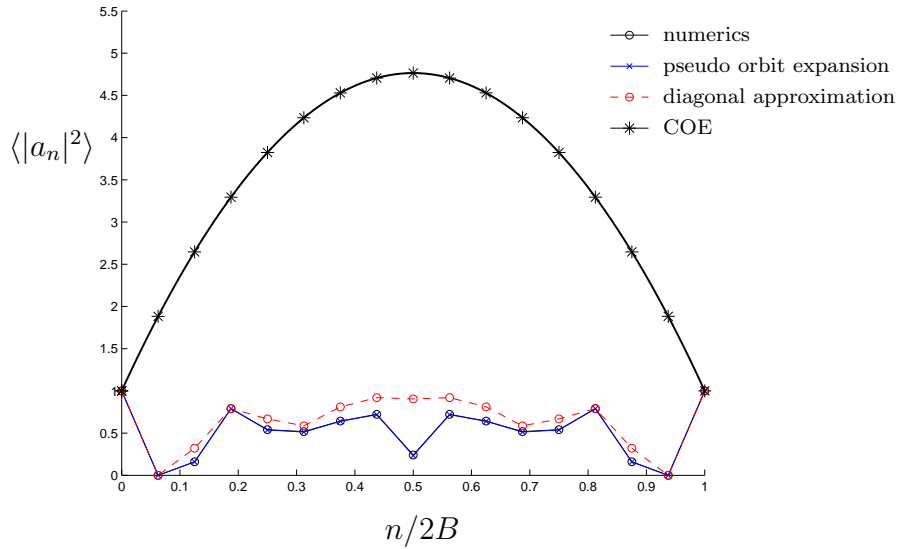}}
    \put(5,0){$n/2B$}
    \put(0,5.1){$\langle |a_n|^2 \rangle$}
    \put(8.7,6.62){\scriptsize numerics}
    \put(8.7,6.21){\scriptsize pseudo orbit expansion}
    \put(8.7,5.8){\scriptsize diagonal approximation}
    \put(8.7,5.39){\scriptsize COE}
    \end{picture}
  \caption{\it Variance of coefficients of the characteristic polynomial for the graph shown in figure \ref{fig:fcs}(b) with random bond lengths chosen uniformly in $[0,1]$.  Numerical results were averaged over $200\,000$ values of $k$ over a range of approximately $1700$ mean level spacings and are compared with the pseudo orbit expansion, diagonal approximation and random matrix theory.}
  \label{fig:variance2}
  \end{center}
\end{figure}

The variance of the coefficients of the characteristic polynomial of a $2B\times 2B$ random matrix were obtained in \cite{p:Hetal96}.
For random matrices from the Circular Orthogonal Ensemble (COE) and Circular Unitary Ensemble (CUE) ($\beta=1$ and $\beta=2$ respectively) the variance of the coefficients are given by,
\begin{equation}\label{eq:RMT coeffs of char poly}
    \langle |a_n|^2 \rangle_\beta = \left\{ \begin{array}{ll}
    1 +\frac{n(2B-n)}{2B+1} & \qquad \beta =1 \\
    1& \qquad \beta=2 \\
    \end{array}
    \right. \ .
\end{equation}
%

In the semiclassical limit the spectrum of generic quantum graphs is seen to approach that of an appropriate random matrix ensembles.  This goes back to the results of Kottos and Smilansky for quantum graphs \cite{p:KS99} which is the graph setting of the well known correspondence of Bohigas, Giannoni and Schmit \cite{p:BGS84}.  The semiclassical limit is the limit of increasing spectral density, which for quantum graphs is the limit of a sequence of graphs with increasing number of bonds \cite{GnuSmi_aip06}.  In \cite{p:BW10} Berkolaiko and Winn formalize the connection between the spectrum of a scattering matrix and the spectrum of the corresponding Laplace operator on the graph.  Of particular relevance to the current discussion, when the spectrum of the graph is distributed like the eigenvalues of a random matrix we also expect that the eigenphases of the scattering matrix on the unit circle will have the same distribution.  When this is the case one must also anticipate that the coefficients of the characteristic polynomial of the graph are distributed according to the random matrix prediction.  As this was not observed it appears that the variance of the coefficients of the characteristic polynomial is more sensitive to the size of the graph used for the numerics than the spectrum itself.
In \cite{p:Tan02} Tanner investigated sequences of binary graphs without time-reversal symmetry where he observes converges to a constant different from the random matrix result and dependent on the classical limit of an underlying Markov process.

\section{Zeta function}\label{sec:zeta}
Let $z=k+\ui t$ with $k,t$ real, following theorem \ref{thm:secular function} we will denote
\begin{equation}\label{eq:defn F}
    F(z)=2 \sum_{\ppo | B_{\ppo} \leq B} (-1)^{m_{\ppo}} A_{\ppo} \cos(z(l_{\ppo}-\gL)-\theta) \, H(B-B_{\ppo})
\end{equation}
so that the secular equation reads $F(z)=0$.   $F(k)$ is an even or odd function of $k$ depending on whether $\theta$ is $0$ or $\pi/2$, so we have a sum of either cosine or sine functions respectively.
Notice that $F(k)$ must be even or odd as the $k$-spectrum is symmetric about zero.
For Neumann vertex conditions in the limit $z\to 0$,
\begin{equation}\label{eq:F z->0}
    F(z)\underset{z\to 0}{\sim} z^{B-V+2} + O( z^{B-V+4}) \ ,
\end{equation}
which comes from the identity $\textrm{dim} \{ \ker (\UI-U(0))\}=B-V+2$ \cite{p:K08}.
We define
\begin{equation}\label{eq:hatF}
    \hat{F}(z)=\frac{F(z)}{z^{B-V+2}} \ .
\end{equation}
Importantly $\hat{F}$ also has roots $k_n$ but does not vanish at zero, $\hat{F}(z)\sim 1+O(z^2)$ as $z\to 0$.  We see that $\hat{F}(-z)=\hat{F}(z)$. To formulate the spectral zeta function we follow a procedure introduced for graphs in \cite{p:HK11} which we only sketch here.  We take the function $\hat{F}(z)$ and use the argument principle, evaluating an integral around a contour $c$ enclosing the positive real axis, see figure \ref{fig:contours}(a).
\begin{equation}\label{eq:contour int}
    \zeta(s,\lambda)=\frac{1}{2\pi \ui} \int_c (z^2+\lambda)^{-s} \frac \ud {\ud z} \log \hat{F}(z) \,  \ud z \ ,
\end{equation}
which converges for $\re s > 1$.

\begin{figure}[!ht]
  \begin{center}
  \setlength{\unitlength}{1cm}
    \begin{picture}(13,5)
    \put(0,0){\includegraphics[width=6cm]{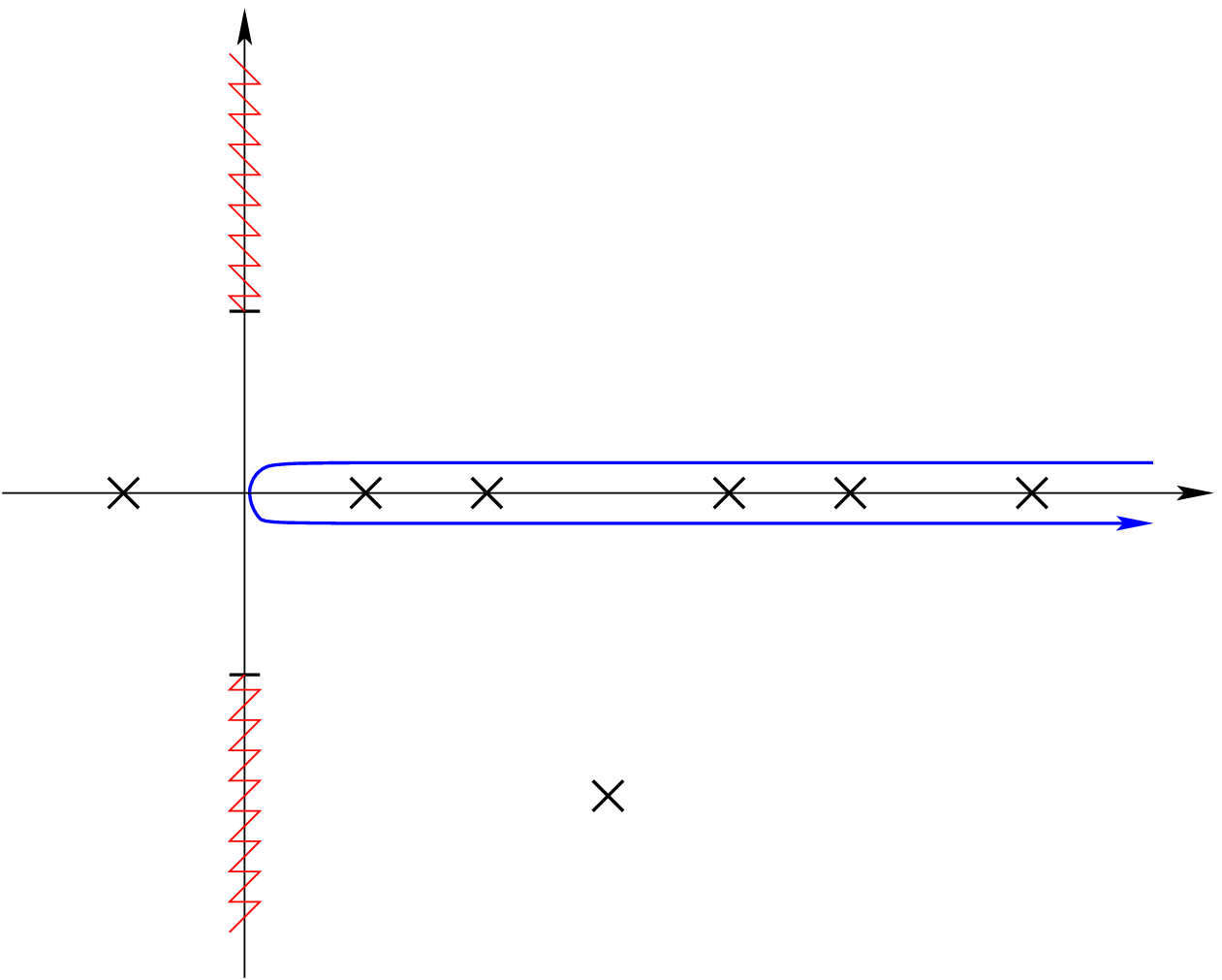}}
    \put(7,0){\includegraphics[width=6cm]{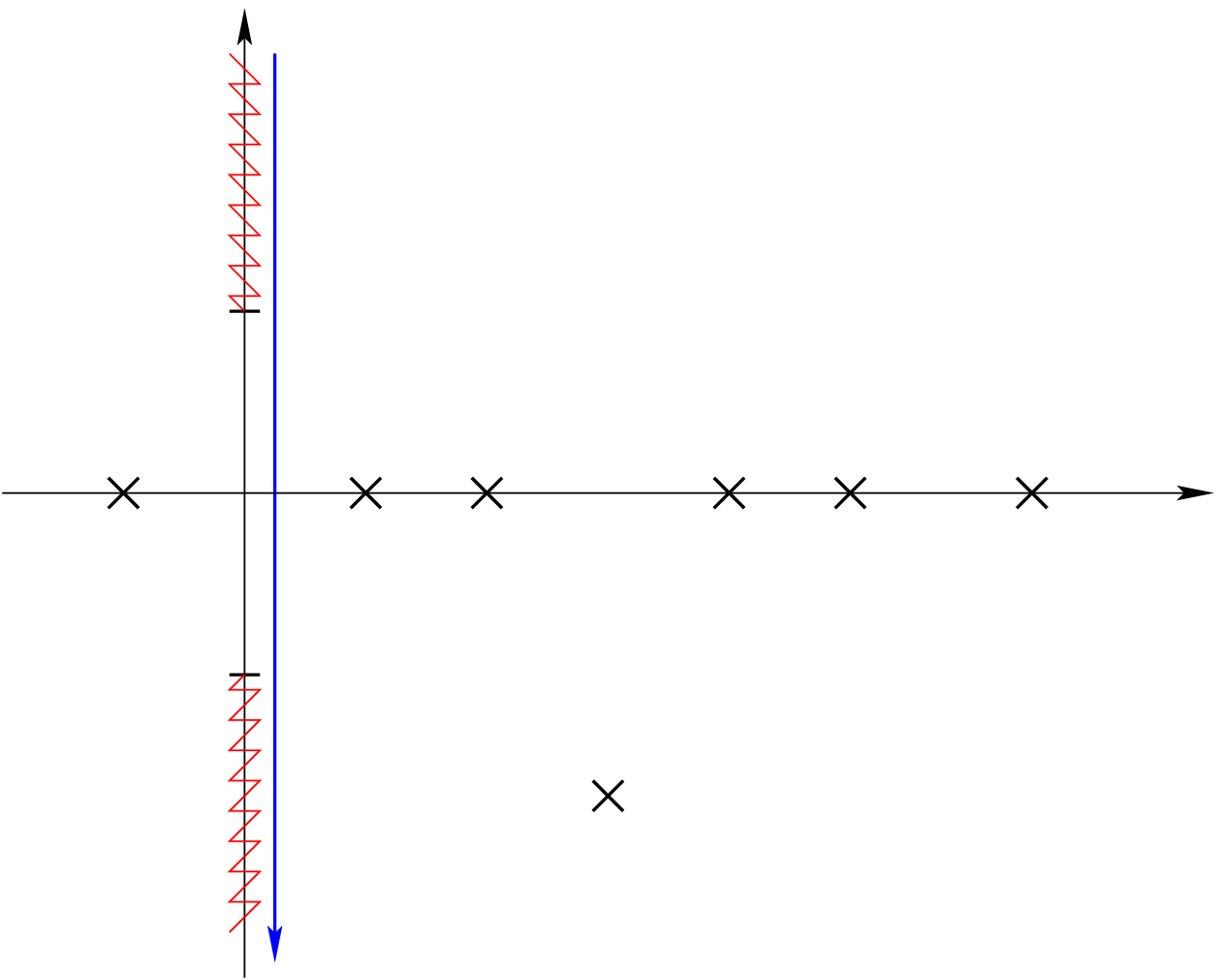}}
    \put(0,4.8){(a)}
    \put(7,4.8){(b)}
    \put(3.3,0.8){zero of $\hat{F}(z)$}
    \put(10.3,0.8){zero of $\hat{F}(z)$}
    \put(-0.1,1.4){$-\ui \sqrt{\lambda}$}
    \put(0.2,3.2){$\ui \sqrt{\lambda}$}
    \put(5.7,1.9){$k$}
    \put(0.8,4.5){$\ui t$}
    \put(12.7,1.9){$k$}
    \put(7.8,4.5){$\ui t$}
    \put(5,1.8){$c$}
    \put(8.6,0.5){$c'$}
    \end{picture}
  \caption{\it Contours used to formulate $\zeta(s,\lambda)$
    (a) before, and (b) after, the contour transformation.
    Branch cuts are represented by hashed lines.}
  \label{fig:contours}
  \end{center}
\end{figure}

We may deform the contour $c$ to $c'$ where the integral is along the imaginary axis, figure \ref{fig:contours}(b).
\begin{equation}\label{eq:zeta c'}
    \zeta(s,\lambda) = \frac{\sin \pi s}{\pi} \int_{\sqrt{\lambda}}^{\infty} (t^{2}-\lambda)^{-s} \frac{\ud}{\ud t}
    \log \hat{F}(\ui t) \,  \ud t
\end{equation}
The segment between $\ui \sqrt{\lambda}$ and $-\ui \sqrt{\lambda}$ does not contribute to the integral due to the symmetry of  $\hat{F}(z)$.  Analysis of the integral shows that the representation of the zeta function (\ref{eq:zeta c'}) converges in the strip $1/2<\re s <1$.  As we will see, the restriction $\re s>1/2$ is due to the $t\to \infty$ behavior of $\hat{F}(\ui t)$.
\begin{equation}\label{eq:F t->inf a}
    \hat{F}(\ui t) =
    \frac{\pm 1}{t^{B-V+2}} \sum_{\ppo | B_{\ppo} \leq B} (-1)^{m_{\ppo}} A_{\ppo} \, (\ue^{t(\gL-l_{\ppo})} \pm \ue^{t(l_{\ppo}-\gL)}) \, H(B-B_{\ppo}) \ ,
\end{equation}
where the choice of sign for $\ue^{t(l_{\ppo}-\gL)}$ depends on whether $\theta=0$ or $\pi/2$.\footnote[1]{The sign of the whole function is not important as multiplying $\hat{F}$ by a constant factor does not change its properties.}
As no directed bond appears in $\ppo$ more than once $l_{\ppo}<\gL$.
Hence
\begin{equation}\label{eq:F t->inf}
    \hat{F}(\ui t) \underset{t\to \infty}{\sim}
    \frac{\ue^{t\gL}}{t^{B-V+2}} \ ,
\end{equation}
which is the contribution of the null pseudo orbit $\overline{0}$; the pseudo orbit of length zero.
Consequently,
\begin{equation}\label{eq:log F t->inf}
   \frac{\ud}{\ud t} \log \hat{F}(\ui t) \underset{t\to \infty}{\sim}
    \gL-\frac{B-V+2}{t} \ ,
\end{equation}
and we see that (\ref{eq:zeta c'}) converges for $\re s>1/2$.  To make the analytic continuation of the zeta function to the left of the line $\re s =1/2$ we add and subtract the $t\to \infty$ asymptotics in the integral.  We thus obtain an expression for $\zeta(s,\lambda)$ valid for $\re s<1$,
\begin{eqnarray}\label{eq:zeta reg}
    \zeta(s,\lambda) = & \frac{\sin \pi s}{\pi} \int_{\sqrt{\lambda}}^{\infty} (t^{2}-\lambda)^{-s} \frac{\ud}{\ud t}
    \log (\hat{F}(\ui t)t^{B-V+2} \ue^{-t\gL}) \,  \ud t
    \nn \\
    & + \gL \frac{\Gamma(s-1/2)}{2\sqrt\pi \Gamma(s)}\lambda^{\frac12-s}
    -\frac{B-V+2}{2} \lambda^{-s} \ .
\end{eqnarray}

\section{Spectral determinants and vacuum energy}\label{sec:trace}
We are now in a position to use the zeta function to formulate the spectral determinant and vacuum energy of the graph as sums over a finite set of pseudo orbits.
\subsection{Spectral determinant}\label{sec:det}
The spectral determinant, formally the product of the shifted eigenvalues (\ref{eq:formal spec det}), was defined as $\gS(\lambda)=\exp(-\zeta'(0,\lambda))$, see \cite{p:HK11}.
From (\ref{eq:zeta reg}) we see
\begin{eqnarray}\label{eq:zeta'}
\fl    \zeta'(0,\lambda) &=  \left[ \log (\hat{F}(\ui t)t^{B-V+2}\ue^{-t\gL}) \right]^{\infty}_{\sqrt{\lambda}} - \gL \sqrt{\lambda}
    +\frac{B-V+2}{2} \log \lambda
    \nn \\
\fl    & =\log A_{\overline{0}}-\log (\hat{F}(\ui \sqrt{\lambda})\lambda^{\frac{B-V+2}{2}}\ue^{-\sqrt{\lambda}\gL}) - \gL \sqrt{\lambda}+\frac{B-V+2}{2} \log \lambda \ .
\end{eqnarray}
Hence the spectral determinant is given as
\begin{equation}\label{eq:graph det}
 \fl   \gS(\lambda)
     = \frac{2}{\lambda^{(B-V+2)/2}}
     \sum_{\ppo | B_{\ppo} \leq B} (-1)^{m_{\ppo}} A_{\ppo} \cosh(\sqrt{\lambda}(l_{\ppo}-\gL)) \, H(B-B_{\ppo}) \ ,
\end{equation}
when $\theta=0$ or the same formula with $\cosh$ replaced by $\sinh$ when $\theta=\pi/2$.
This concisely formulates the spectral determinant of the graph in terms of a finite set of dynamical quantities.

\subsection{Vacuum energy}\label{sec:vac}
A second application of the pseudo orbit approach is the zeta-regularized vacuum energy.  The vacuum energy, formally the sum of the square roots of the eigenvalues (\ref{eq:formal Ec}), is defined as
\begin{equation}\label{eq:defn Ec}
    \E_c=\frac{1}{2} \zeta ( -1/2,0 ) \ .
\end{equation}
Were $\zeta(s,0)$ to have pole at $s=-1/2$ a more general definition would be required, see e.g.
\cite{b:Kir:SpectralFunctions}; however, we will see that there is no pole in the present case.
%
%
The Casimir force on a bond $b$ is given by $-\frac{\partial }{\partial L_b} \mathcal{E}_c$.
%
%

The zeta function for $\lambda=0$ is,
\begin{equation}\label{eq:zeta c' E=0}
    \zeta(s,0) = \frac{\sin \pi s}{\pi} \int_{0}^{\infty} t^{-2s} \frac{\ud}{\ud t}
    \log \hat{F}(\ui t) \,  \ud t \ ,
\end{equation}
in the strip $1/2<\re s<1$.  Splitting the integral at $t=1$
\begin{equation}\label{eq:zeta c' E=0 split}
    \zeta(s,0) = \frac{\sin \pi s}{\pi} \left( \int_{0}^{1} t^{-2s} \frac{\ud}{\ud t}
    \log \hat{F}(\ui t) \,  \ud t + \int_{1}^{\infty} t^{-2s} \frac{\ud}{\ud t}
    \log \hat{F}(\ui t) \,  \ud t \right) \ ,
\end{equation}
and the restriction to $\re s>1/2$ comes from the $t\to \infty$ behavior in the second integral.
Adding and subtracting the asymptotic behavior (\ref{eq:F t->inf}) in the second integral (as in (\ref{eq:zeta reg})) we obtain,
\begin{eqnarray}\label{eq:zeta c' E=0 2}
\fl    \zeta(s,0) =& \frac{\sin \pi s}{\pi} \int_{0}^{1} t^{-2s} \frac{\ud}{\ud t}
    \log \hat{F}(\ui t) \,  \ud t  + \frac{\sin \pi s}{\pi} \int_{1}^{\infty} t^{-2s} \frac{\ud}{\ud t}
    \log (\hat{F}(\ui t) t^{B-V+2} \ue^{-t\gL}) \,  \ud t
    \nn \\
\fl    & + \frac{\sin \pi s}{\pi} \left( \frac{\gL}{2s-1} +\frac{V-B-2}{2s} \right)\ ,
\end{eqnarray}
which is now valid for $\re s<1$.  Using this representation we obtain
\begin{eqnarray}\label{eq:FP zeta(-1/2,0)}
\fl   \mathcal{E}_c =  &-\frac{1}{2\pi} \int_{0}^{1} t \frac{\ud}{\ud t}
    \log \hat{F}(\ui t) \,  \ud t  +\frac{\gL}{4\pi} +\frac{V-B-2}{2\pi}\nn  \\
\fl    & +\frac{1}{2\pi} \sum_{\ppo|B_{\ppo} \leq B} (-1)^{m_{\ppo}} A_{\ppo}
    \int_1^\infty \frac{ tl_{\ppo} \ue^{-tl_{\ppo}} - t(2\gL-l_{\ppo}) \ue^{-t(2\gL-l_{\ppo})}}{
    \sum_{\ppo|B_{\ppo} \leq B} (-1)^{m_{\ppo}} A_{\ppo}
    (\ue^{-tl_{\ppo}} +\ue^{-t(2\gL-l_{\ppo})})}\, \ud t\ .
\end{eqnarray}
This is an exact trace formula for the vacuum energy where the lengths of orbits included in the sum are truncated at a finite value.  For a particular graph the integrals can be evaluated numerically.

\section{Conclusions}\label{sec:conclusions}
We have investigated a pseudo orbit approach to spectral questions on quantum graphs.  The possibility of a pseudo orbit expansion
of the secular function of a graph was described in \cite{p:KS99}.  Here, however, we have presented the explicit formula.  As quantum graphs provide a system where the spectral analysis can be carried out exactly, they serve as an ideal example in which to investigate the various approaches to cancelations which appear in pseudo orbit expansions.  To this end while the permutation expansion of the determinant provides a neat direct route to the expansion using irreducible pseudo orbits, we also describe other methods of deriving this formula in the appendix.  In these cases a sum over pseudo orbits is seen to contain substantial cancelations when pseudo orbits are grouped appropriately.

Having obtained the expansion of the secular function in terms of pseudo orbits we demonstrated how this can be applied to three spectral quantities, the variance of the coefficients of the characteristic polynomial, the spectral determinant and the vacuum energy.
We provide an exact pseudo orbit expansion for these quantities using only dynamical information for a finite set of short periodic orbits
whose topological length is at most the length of the graph. In the case of the vacuum energy it is particularly interesting to note that such a quantity can be obtained from a finite set of dynamical properties.
In more complex physically realistic systems it is often hard to calculate the vacuum energy.  If it is possible to extend these techniques to compute the vacuum energy of higher dimensional systems one may be able to describe the vacuum energy precisely using orbits with a finite cut off by length.  This would seem to offer the possibility of precise numerical evaluation of the vacuum energy of general systems.  It is a hard problem to provide an accurate numerical evaluation of the vacuum energy, such calculations are often based on the proximity force approximation where a rigorous justification is lacking.

\ack{}{The authors would like to thank Lilian Matthiesen for the elegant proof of the combinatorial identity at the end of the appendix, Eric Akkermans, Jon Keating, Uzy Smilansky and Nina Snaith for helpful discussions and the anonymous referees for their comments.

  JMH would like to thank Bristol University for their hospitality during his sabbatical when the work was carried out. RB is supported by EPSRC, grant number EP/H028803/1.
  JMH was supported by the Baylor University research leave and summer sabbatical programs.
  }

\appendix
\section*{Appendix - More on cancelation mechanisms}
\setcounter{section}{1}

The appendix is dedicated to a description of the various canceling
mechanisms which play a role in the derivation of pseudo orbits expansions.
A combination of some of these schemes provides an alternative
proof of theorem \ref{thm:a_n_expansion}. Some cancelation mechanisms
were already presented in the literature in various contexts (not
necessarily for quantum graphs), as we will mention. However, not
all of them were fully developed to present the current pseudo orbits
expansion. The purpose of bringing them together here is to give
a state of the art description of the various mechanisms and their
interrelations.

\subsection{From an infinite expansion to a finite one}

We start from the characteristic polynomial of $U\left(k\right)$
and use the matrix identity $\det\left(\UI-U(k)\right)=\exp\textrm{Tr}\ln\left(\UI-U(k)\right)=\exp\textrm{Tr}\left(-\sum_{j=1}^{\infty}\frac{U(k)^{j}}{j}\, \right)$
to obtain\footnote[1]{This holds for $k+\ui \epsilon$ which is to be understood in the limit $\epsilon \to 0$.}

\begin{equation}
F_{\xi}(k)=\xi^{2B}\det(\UI-\xi^{-1}U(k))=\xi^{2B}\exp\left(-\sum_{j=1}^{\infty}\frac{\xi^{-j}T_{j}}{j}\right),\label{eq:characteristic_polynomial_1}
\end{equation}
 where we use the notation $T_{j}:=\mbox{Tr}[U(k)^{j}]$.

The traces $T_{j}$ can be expressed in terms of periodic orbits
\begin{equation}
T_{j}=\sum_{\gamma |\, B_{\gamma}=j}\frac{j}{r_{\gamma}}A_{\gamma}\ue^{\ui l_{\gamma}k},\label{eq:trace_in_terms_of_orbits}
\end{equation}
see notation in section \ref{sec:expansion_for_secular_function}.
Inserting (\ref{eq:trace_in_terms_of_orbits}) in (\ref{eq:characteristic_polynomial_1})
we see,
\begin{equation}\label{eq:part way}
    F_{\xi}(k) = \xi^{2B}\exp\left(-\sum_{j=1}^{\infty}\sum_{\gamma |\, B_{\gamma}=j}\xi^{-j}\frac{A_{\gamma}\ue^{\ui kl_{\gamma}}}{r_{\gamma}}\right) \ .
\end{equation}
Fixing a primitive periodic orbit $\pi$ and summing over its repetitions,
\begin{eqnarray}
F_{\xi}(k)
 & = & \xi^{2B}\exp\left(-\sum_{\pi\,\textrm{primitive}}
 \sum_{r=1}^{\infty}\xi^{-rB_{\pi}}\frac{A_{\pi}^{r}\ue^{\ui rkl_{\pi}}}{r}\right)\nonumber \\
 & = & \xi^{2B}\exp\left(\sum_{\pi\,\textrm{primitive}}
 \ln\left(1-\xi^{-B_{\pi}}A_{\pi}\ue^{\ui kl_{\pi}}\right)\right)\\
 & = & \xi^{2B}\prod_{\pi\,\textrm{primitive}}
 \left(1-\xi^{-B_{\pi}}A_{\pi}\ue^{\ui kl_{\pi}}\right)\ . \nonumber
\end{eqnarray}
Expanding the last line above, one would get an infinite sum,
each of whose terms is the contribution of some collection of primitive periodic
orbits. Namely, each term would correspond to a pseudo orbit which
consists of only primitive orbits; a \emph{primitive pseudo orbit}.
This is not to be confused with an irreducible pseudo orbit
see section \ref{sec:expansion_for_secular_function} - which
form a subset of the primitive ones. Keeping in mind that the characteristic
polynomial has only finite powers of $\xi$, $F_{\xi}\left(k\right)=\xi^{2B}\sum_{n=0}^{2B}a_{n}\xi^{-n}$,
we see that the contribution of many of these pseudo orbits cancel.
The remaining pseudo orbits give the following expression for the
$n^{\textrm{th}}$ coefficient of the polynomial:
\begin{equation}
a_{n}=\sum_{\overset{
\po\textrm{ primitive}}{
B_{\po}=n}}
(-1)^{m_{\po}}A_{\po}\ue^{\ui kl_{\po}} \ . \label{eq:cancellation_mechanism_1_result}
\end{equation}

\subsection{Cancelations due to a combinatorial identity}\label{subsec:cancelations}

A pseudo orbit expansion for quantum graphs appeared in \cite{p:KS99}.
The formula for evaluating the pseudo orbits contributions was given
there, but the coefficients were not computed explicitly. We compute
those coefficients and show that many cancelations occur due to a
combinatorial identity, leading again to the formula (\ref{eq:cancellation_mechanism_1_result}).

In the following we will relate pseudo orbits to partitions of
an integer.
A \emph{partition} of some positive integer, $n$, is an ordered set
$P=\left(m_{1},m_{2},\ldots,m_{n}\right)$ such that $\sum_{j=1}^{n}j\, m_{j}=n$.
In other words, it is a way to present $n$ as a sum of all positive
integers smaller than $n$ (each can be taken once, more than once,
or not at all).
The set of all partitions of $n$ is denoted by $\mathcal{P}_{n}$.

We now relate a partition of $n$ to a certain class of pseudo
orbits. Let $P=\left(m_{1},m_{2},\ldots,m_{n}\right)\in\mathcal{P}_{n}$ and
denote by $\tilde{\Gamma}_{P}$ the set of all pseudo orbits, $\tilde{\gamma}$,
such that for all $1\leq j\leq n$, $\tilde{\gamma}$
contains $m_{j}$ periodic orbits of topological length $j$.
Hence for $\po\in \tilde{\Gamma}_{P}$ we have $B_{\po}=n$.

We begin the proof by employing Newton's relations which connect coefficients of the characteristic polynomial
$a_{n}$ to traces of powers of $U\left(k\right)$ (see e.g. \cite{b:H:QSC}). One
form for these identities is
\begin{equation}
a_{n}=-\frac{1}{n}\sum_{j=0}^{n-1}\tr\left(U^{n-j}\left(k\right)\right)a_{j}\ , \label{eq:newton_identities}
\end{equation}
with $a_{0}=1$.
We use this recursive relation to give a closed expression for the
coefficients in terms of the traces.
\begin{equation}
a_{n}=\sum_{P\in\mathcal{P}_{n}}\frac{1}{m_{1}!}\left(-\frac{T_{1}}{1}\right)^{m_{1}}\frac{1}{m_{2}!}\left(-\frac{T_{2}}{2}\right)^{m_{2}}\cdot\ldots\cdot\frac{1}{m_{n}!}\left(-\frac{T_{n}}{n}\right)^{m_{n}},\label{eq:expression_for_a_n}
\end{equation}
 where $T_{j}:=\tr\left(U^{j}\left(k\right)\right)$.
A similar expression appears in \cite{b:H:QSC,p:KS99}. In what follows
we show that (\ref{eq:expression_for_a_n}) can be expanded as a sum over pseudo
orbits and that careful book-keeping allows a significant
number of cancelations.

The traces $T_{j}$ can be presented as sums over periodic orbits on
the graph as in (\ref{eq:trace_in_terms_of_orbits}).
Substituting in (\ref{eq:expression_for_a_n})
we obtain
\begin{eqnarray}
 a_{n}&=\sum_{P\in\mathcal{P}_{n}}\prod_{j=1}^{n}\frac{1}{m_{j}!}\left(-\frac{T_{j}}{j}\right)^{m_{j}} \nn \\ &=\sum_{P\in\mathcal{P}_{n}}  \prod_{j=1}^{n}\left\{ \frac{1}{m_{j}!}\left(-\frac{1}{j}\sum_{\gamma |\, B_{\gamma}=j}\frac{B_{\gamma}}{r_{\gamma}}A_{\gamma}\ue^{\ui kl_{\gamma}}\right)^{m_{j}}\right\} \nonumber \\
 & =\sum_{P\in\mathcal{P}_{n}} \prod_{j=1}^{n}\left\{ \frac{\left(-1\right)^{m_{j}}}{m_{j}!}\left(\sum_{\gamma |\, B_{\gamma}=j}\frac{1}{r_{\gamma}}A_{\gamma}\ue^{\ui kl_{\gamma}}\right)^{m_{j}}\right\} .\label{eq:expression_for_a_n_2}
\end{eqnarray}
One may expand the second sum and the product in the expression above.
This results in a sum, each of whose terms corresponds to a product
over periodic orbits. Such a product can be identified with a pseudo
orbit. More specifically,
\begin{equation}
a_{n}=\sum_{\po |B_{\po}=n}C_{\po}A_{\po}\ue^{\ui kl_{\po}},\label{eq:pseudo_orbits_expansion_before_cancellations}
\end{equation}
where $C_{\po}$ is a coefficient that consists of combinatorial
factors and repetition numbers. There are many cancelations in this
sum and we will show that eventually only primitive pseudo orbits
remain.

Let us fix some primitive periodic orbit, $\pi$. Consider all
pseudo orbits $\po$ with $B_{\po}=n$ for some fixed $n$, which contain only copies of $\pi$ and
its possible repetitions. For example, such a
pseudo orbit might be $\po=\left\{ 5\pi,2\pi^{3}\right\} $, which
means that $\po$ contains seven periodic orbits; five copies of $\pi$
and two copies of a periodic orbit where $\pi$ is repeated three times.
 For all such pseudo orbits $B_{\po}=n$ is a multiple
of $B_{\pi}$, and we denote $q:=B_{\po}/ B_{\pi}$.
Summing the terms in (\ref{eq:pseudo_orbits_expansion_before_cancellations}) over this collection of
pseudo orbits and reading the expression for $C_{\po}$ from
(\ref{eq:expression_for_a_n_2}), we obtain

\begin{equation}
\left\{\sum_{Q\in\mathcal{P}_{q}} \left(\prod_{r=1}^{q}\left(-1\right)^{n_{r}}\frac{1}{n_{r}!}\left(\frac{1}{r}\right)^{n_{r}}\right)\right\}\left(A_{\pi}\ue^{\ui kl_{\pi}}\right)^{q} ,\label{eq:camcellations-1}
\end{equation}
where the sum is over partitions of $q$, $Q=\left(n_{1},\ldots,n_{q}\right)$.
To obtain (\ref{eq:camcellations-1}) observe that;
\begin{enumerate}
\item We only sum over pseudo orbits belonging to $n$-partitions in (\ref{eq:expression_for_a_n_2})
with $m_{j}=0$ for all $j$ values which are not multiple of $B_{\pi}$.
So we may enumerate these by considering the partitions of $q=B_{\po}/B_{\pi}$ where $n_r=m_{rB_{\pi}}$.
\item Each partition of $q$ determines a unique pseudo orbit constructed only from $\pi$. Denoting a $q$-partition by $Q=\left(n_{1},\ldots,n_{q}\right)$,
the repetition numbers of the orbits within our pseudo
orbit are just the indices of the $n_{r}$'s.
\item All these pseudo orbits are weighted by the same factor, $\left(A_{\pi}\ue^{\ui kl_{\pi}}\right)^{q}$.
\end{enumerate}
Lemma \ref{lem:combinatorial_cancellation} will show that the sum
in (\ref{eq:camcellations-1}) is zero for all $q>1$. This
implies that the only surviving contribution comes from the pseudo orbit that contains
$\pi$ alone.
Furthermore, we can extend this cancelation argument as follows. Fix $n$ and let $\pi_1$ and $\pi_2$ be two primitive periodic orbits, such that $n=q_1B_{\pi_1}+q_2B_{\pi_2}$. We now examine all pseudo orbits with $n$ bonds which contain only copies of $\pi_1$, $\pi_2$, and their
 repetitions. Collecting the contributions of all such orbits from (\ref{eq:expression_for_a_n_2}), we arrive at
\begin{equation*}
\fl \sum_{Q_1\in\mathcal{P}_{q_1}} \left(
\prod_{r=1}^{q_1}\left(-1\right)^{n_{r}}\frac{1}{n_{r}!}\left(\frac{1}{r}\right)^{n_{r}} \right)
\sum_{Q_2\in\mathcal{P}_{q_2}} \left(
\prod_{r=1}^{q_2}\left(-1\right)^{n_{r}}\frac{1}{n_{r}!}\left(\frac{1}{r}\right)^{n_{r}}\right) 
(A_{\pi_1}\ue^{\ui kl_{\pi_1}})^{q_1}
(A_{\pi_2}\ue^{\ui kl_{\pi_2}})^{q_2}.\label{eq:camcellations-2}
\end{equation*}
 Applying lemma \ref{lem:combinatorial_cancellation} again, the sum above vanishes if either $q_1>1$ or $q_2>1$ and hence the only contribution that does not vanish is for the primitive pseudo orbit $\{\pi_1,\pi_2\}$.
Iterating this argument we conclude that all contributions from non-primitive pseudo orbits
cancel. Hence
\begin{equation}
a_{n}=\sum_{\overset{\po\textrm{ primitive}}{B_{\po}=n}}(-1)^{m_{\po}}A_{\po}\ue^{\ui kl_{\po}},\label{eq:cancellation_mechanism_2_result}
\end{equation}
which agrees with (\ref{eq:cancellation_mechanism_1_result})
obtained using the first canceling mechanism.

\begin{lemma}\label{lem:combinatorial_cancellation}
For all $q\in \bbN$ with $q>1$,
\[
\sum_{P\in\mathcal{P}_{q}}\prod_{r=1}^{q}\left(-1\right)^{n_{r}}\frac{1}{n_{r}!}\left(\frac{1}{r}\right)^{n_{r}}=0.
\]

\end{lemma}

\begin{proof}

The expression we wish to vanish is equal to the coefficient
multiplying $x^{q}$ in the following generating function,
\begin{eqnarray*}
\prod_{r=1}^{\infty}\sum_{n_{r}=0}^{\infty}\frac{1}{n_{r}!}\left(-\frac{x^{r}}{r}\right)^{n_{r}} & = & \prod_{r=1}^{\infty}\exp\left(-\frac{x^{r}}{r}\right)\\
 & = & \exp\left(-\sum_{r=1}^{\infty}\frac{x^{r}}{r}\right)\\
 & = & \exp\left(\ln\left(1-x\right)\right)\\
 & = & 1-x \ .
\end{eqnarray*}

\end{proof}

\subsection{Further diagrammatic cancelations}

Comparing the results of the two canceling mechanisms
(\ref{eq:cancellation_mechanism_1_result}) and (\ref{eq:cancellation_mechanism_2_result})
with theorem \ref{thm:a_n_expansion},
we see the difference is between a sum over all primitive pseudo
orbits and a sum only over the irreducible ones. More specifically,
in theorem \ref{thm:a_n_expansion} we do not sum over primitive pseudo orbits which
contain the same directed bond more than once. It is possible to show
that the contribution of all such periodic orbits cancel each other using a diagrammatic
approach, as was shown very nicely in \cite{AkkComDebMonTex_ap00}.
One should note that in \cite{AkkComDebMonTex_ap00}
the diagrammatic cancelations are presented assuming each vertex scattering matrix, $\sigma^{\left(v\right)}$,
has all its off diagonal entries equal.
However, their diagrammatic cancelations which can be used to transform (\ref{eq:cancellation_mechanism_2_result})
into (\ref{eq:a_n_expansion}) do not require this assumption.
We also note that a similar cancelation mechanism appears in
\cite{Bog_chaos92}. However, the derivation there is done for the
dynamical zeta function and the cancelations occur between orbits
which pass through the same cell in the Poincar\'e surface of section.

\section*{References}{\bibliographystyle{plain}
\bibliography{Pseudo_Orbits_and_Riemann_Siegel,Trace_Formulae,reffs-jmh,Qunatum_Graphs_Introduction}
}

\end{document}